\documentclass[11pt,fleqn]{article}
\usepackage{amsmath,amssymb,graphicx,amsthm,dsfont}
\usepackage{color,soul}
\usepackage{xspace}
\usepackage{enumerate}
\usepackage{mathtools}
\usepackage{bbm}
\usepackage[numbers]{natbib}

\usepackage{booktabs}

\usepackage{subcaption}
\usepackage[noend,ruled,linesnumbered]{algorithm2e} 
\usepackage{xifthen}

\usepackage[colorlinks,citecolor=green!60!black,linkcolor=red!60!black,pagebackref]{hyperref}

\usepackage{tikz}
\usetikzlibrary{decorations,arrows,petri,topaths,backgrounds,shapes,positioning,fit,calc,decorations.pathreplacing,patterns,intersections,decorations.pathmorphing,matrix}

\tikzstyle{thickline} = [line width=1.8pt]

\makeatletter
\newcommand{\gettikzxy}[3]{%
  \tikz@scan@one@point\pgfutil@firstofone#1\relax
  \edef#2{\the\pgf@x}%
  \edef#3{\the\pgf@y}%
}
\makeatother

\usepackage{cleveref}

\newtheorem{lemma}{Lemma}
\newtheorem{observation}{Observation}
\newtheorem{proposition}{Proposition}

\newtheorem{theorem}{Theorem}

\theoremstyle{definition}
\newtheorem{definition}{Definition}
\newtheorem{example}{Example}

\crefname{table}{Table}{Tables}
\crefname{figure}{Figure}{Figures}
\crefname{theorem}{Theorem}{Theorems}
\crefname{definition}{Definition}{Definitions}
\crefname{corollary}{Corollary}{Corollaries}
\crefname{observation}{Observation}{Observations}
\crefname{lemma}{Lemma}{Lemmas}
\crefname{example}{Example}{Examples}
\crefname{reduction}{Reduction}{Reductions}
\crefname{construction}{Construction}{Constructions}
\crefname{subsection}{Subsection}{Subsections}
\crefname{section}{Section}{Sections}
\crefname{proposition}{Proposition}{Propositions}
\crefname{algorithm}{Algorithm}{Algorithms}
\crefname{claim}{Claim}{Claims}

\newcommand{\rz}{{\mathbb{R}}}

\DeclareMathOperator*{\argmin}{arg\,min}

\newcommand{\acon}{\mbox{$\alpha$}}

\newcommand{\ccon}{\mbox{$\gamma$}}
\newcommand{\dcon}{\mbox{$\delta$}}

\newcommand{\ppp}{{\cal P}}
\newcommand{\vvv}{{\cal V}}

\newcommand{\aaa}{{\cal A}}
\newcommand{\rrr}{{\cal R}}

\newcommand{\worstinconsistentconfig}{worst-diverse configuration\xspace}

\newcommand{\worst}{\mathsf{worst}}

\newcommand{\ode}{one-dimensional Euclidean\xspace}
\newcommand{\Ode}{One-dimensional Euclidean\xspace}
\newcommand{\oder}{one-dimensional Euclidean representation\xspace}
\newcommand{\Oder}{One-dimensional Euclidean representation\xspace}

\newcommand{\scp}{single-crossing property\xspace}

\newcommand{\spp}{single-peaked property\xspace}

\newcommand{\dEuclid}[1][]{{\ifthenelse{\equal{#1}{}}{$d$}{$#1$}-Euclidean}\xspace}

\newcommand{\inner}{\mathsf{inner}}

\newcommand{\pref}{\ensuremath{\succ}}

\usepackage{diagbox}

\SetKwFunction{VoterPositionSearch}{VoterPositionSearch}
\SetKwFunction{SPSearch}{SPSearch}
\SetKwFunction{Embed}{Embed}
\SetKwFunction{Refine}{Refine}
\SetKwFunction{Fallback}{Fallback}
\SetKwInput{kwInit}{Init}

\SetCommentSty{mycommfont}

\begin{document}
\sloppy

\title{{\bf Small One-Dimensional Euclidean Preference Profiles}}
\author{Jiehua Chen \and Sven Grottke}

\maketitle

\begin{abstract}
  We characterize \ode preference profiles with a small number of alternatives and voters.  
  In particular, we show the following. 
  \begin{itemize}
  \item Every preference profile with up to two voters is \ode if and only if it is single-peaked.
  \item Every preference profile with up to five alternatives is \ode if and only if it is single-peaked and single-crossing.
  \end{itemize}
  By \cite{ChePruWoe2017}, we thus obtain that the smallest single-peaked and single-crossing preference profiles that are \emph{not} \ode consist of three voters and six alternatives.

  \medskip
  
  \noindent \textbf{JEL classification
    D81
    D72
  }
\end{abstract}

\section{Introduction}\label{sec:intro}

The \ode preference domain (also known as the unidimensional unfolding domain) is a spatial model of structured preferences which originates from economics~\cite{Hotelling1929,Downs1957}, political sciences~\cite{Stokes1963,BraJonKil2002,BogLas2007}, and psychology~\cite{Coombs1964,BorGroMai2018}.
In this domain, the alternatives and the voters are points in a one-dimensional space, \emph{i.e.}\ on the real line,
such that the preference of each voter towards an alternative decreases as the Euclidean distance between their points increases.

\Ode preferences are necessarily single-peaked~\cite{Black1948} and single-crossing~\cite{Roberts1977} as proven by~\citet{Coombs1964,DoiFal1994,ChePruWoe2017}.
The reverse, however, does not hold.
In his work, \citet{Coombs1964} provided a sample preference profile with 16 voters and 6 alternatives that is single-peaked and single-crossing, but \emph{not} \ode.
This counterexample appears to be quite large for real world scenarios.
For instance, in rank aggregation or winner determination elections, one typically either has few alternatives to begin with, or may consolidate first make a shortlist of only a few alternatives out of many, which will be considered for the final decision.
There are also settings where only a few voters are involved, as for instance in a hiring committee or when planning holidays for a family.  
Hence, a natural question arising in the context of \ode preferences is
whether for profiles with less than 16 voters or less than 6 alternatives, being single-peaked and single-crossing is sufficient to guarantee a \ode embedding. 
In other words, we are interested in the following question: Are there tight
upper bounds on the number of alternatives or voters such that profiles within
these bounds are \ode as long as they are single-peaked and single-crossing?
Recently, \citet{ChePruWoe2017} provided a single-peaked and single-crossing profile with three voters and six alternatives that is not \ode.
In this paper, we show that this counterexample is indeed minimal in terms of the number of voters and the number of alternatives.
In terms of the number of voters, we provide an algorithm that constructs a \ode embedding for any single-peaked preference profile with two voters (see \cref{alg:Euclid-embedding}).
As for the number of alternatives, we show via computer program that all single-peaked and single-crossing preferences with up to five alternatives are \ode (see \cref{thm:5-alts-sp+sc=euclid}).
We refer to the work of \citet{BreCheWoe2016} and the literature cited there for further discussion of the single-peaked and the single-crossing preference domains.

\paragraph{Paper outline.}
In \cref{sec:defi}, we introduce necessary definitions, including single-peaked and single-crossing preferences, and the \oder.
We also discuss some fundamental observations regarding these domain restrictions.
In \cref{sec:main-result-1}, we formulate our first main result in~\cref{thm:Euclidean-relationship,thm:two-votes-SP=Euclidean}.
We prove this result by providing an algorithm (see \cref{alg:Euclid-embedding}) that constructs a \ode embedding for any two preference orders which are single-peaked.
At the end of the section, we provide an example to illustrate \cref{alg:Euclid-embedding} (see \cref{ex:algorithm}).
In \cref{sec:main-result-2}, we provide our second main result by describing the computer program that finds all possible preference profiles with up to five alternatives that are both single-peaked and single-crossing, and uses the publicly available CPLEX solver to provide a \ode embedding for each of theser profiles (see \cref{thm:5-alts-sp+sc=euclid}). The code and the embeddings for all produced profiles are available from~\url{https://tubcloud.tu-berlin.de/s/rSNKkm8dtPkRKnE} and \url{https://tubcloud.tu-berlin.de/s/ArdQzFd8J6L5YFN}, respectively.

\section{Definitions and notations}
\label{sec:defi}

Let $\aaa\coloneqq \{1,\ldots,m\}$ be a set of alternatives.
A \emph{preference order}~$\pref$ over $\aaa$ is a linear order over $\aaa$; a linear order is a binary relation which is total, irreflexive, asymmetric, and transitive.
Given a preference order~$\pref$,
we use $\succeq$ to denote the binary relation which includes $\pref$ and preserves the reflexivity,
\emph{i.e.}~$\succeq \coloneqq \succ\!\cup \{(a,a)\mid a\in \aaa\}$.
An alternative~$c$ is \emph{the most preferred alternative in $\pref$}
if for each alternative~$b\in \aaa$ it holds that $a \succeq b$.
For two distinct alternatives~$a$ and $b$, the relation $a\pref b$ means that $a$ is strictly preferred to (or in other words, ranked higher than) $b$.
The notion~$\{a, b\}\pref_i c$ means that both $a$ and $b$ are strictly preferred to (or in other words, ranked higher than) $c$, but the preference relation between alternatives~$a$ and $b$ is arbitrary but unique.

A \emph{preference profile}~$\ppp$ specifies the preference orders of some voters over some alternatives.
Formally, $\ppp \coloneqq (\aaa, \vvv, \rrr \coloneqq (\pref_1, \ldots, \pref_n))$,
where $\aaa$ denotes the set of $m$ alternatives,
$\vvv$ denotes the set of $n$~voters,
and $\rrr$ is a collection of $n$ preference orders
such that each voter~$v_i\in \vvv$ ranks the alternatives according to the preference order~$\pref_i$ on $\aaa$.
We also assume that no two voters in a preference profile have the same preference orders.


\subsection{Single-peaked preferences}

The single-peaked property was introduced by \citet{Black1958} and has since been studied extensively.

\begin{definition}[\spp] \ \\
  A preference order~$\pref$ on a set~$\aaa$ of alternatives is
  \emph{single-peaked} with respect to a linear order~$\rhd$ of
  the alternatives if for its most preferred alternative~$a^*$ and for each two
  distinct alternatives~$b,c \in \aaa\setminus \{a^*\}$ it holds that
    \begin{align*}
  \text{ if } c \rhd b \rhd a^* \text{ or } a^*\rhd b \rhd c,  \text{then }  b \pref c.
  \end{align*}
  A preference profile with voter set~$\vvv$,
  is \emph{single-peaked} if there is a linear order~$\rhd$ of alternatives such that
  the preference order of each voter from $\vvv$ is single-peaked with respect to~$\rhd$.
\end{definition}

\noindent Slightly abusing the terminology, we say that two preference orders are \emph{single-peaked} if there is a linear order with respect to which each of these two preference orders is single-peaked.

The single-peaked property can be characterized by two forbidden subprofiles, \worstinconsistentconfig{}s and \acon-configurations~\cite{BaHa2011}.
The former is defined on three preference orders while the latter is defined on two preference orders.
By the \acon-configuration, for two arbitrary preference orders~$\pref_1$ and $\pref_2$,
we can observe the following.

\begin{lemma}\label{lem:two-sp-alpha}
  Two preference orders, denoted as $\pref_1$ and $\pref_2$, on the set~$\aaa$
  are single-peaked if and only if 
  for all four distinct alternatives~$x,y,z,w\in \aaa$ such that $x\pref_1 y \pref_1 z$ and $z\pref_2 y \pref_2 x$ 
  it holds that $y \pref_1 w$ or $y \pref_2 w$.
\end{lemma}

\subsection{Single-crossing property}

Single-crossing profiles date back to the seventies, when
\citet{Mirrlees1971} and \citet{Roberts1977} observed that voters
voting on income taxation may form a linear order such that between each
two tax rates, the voters along the order either all agree on the relative
positions of both rates, or there is one spot where the voters switch from
preferring one rate to preferring the other rate.

\begin{definition}[\scp]\label{def:sc}\ \\ 
 A linear order of voters is \emph{single-crossing with respect to
    a pair~$\{a,b\}$ of alternatives}, 
  if there is at most one voter in this order such that 
  all voters ordered ahead of this voter strictly prefer~$a$ to $b$, and all
  voters \emph{not} ordered ahead of this voters strictly prefer~$b$ to $a$.

  A linear order of voters is a \emph{single-crossing} order, 
  if it is single-crossing
  with respect to every possible pair of alternatives. 
  A preference profile is
  \emph{single-crossing} if it allows a single-crossing order of the voters.
\end{definition}

The single-crossing property can be characterized by two forbidden
subprofiles, \ccon-configurations and \dcon-configurations~\cite{BreCheWoe2013a}.

\subsection{\Oder}

\begin{definition}[\oder]\label{def:euclid}\ \\
  Let $\ppp \coloneqq (\aaa, \vvv\coloneqq\{v_1, \ldots,  v_n\}, \rrr\coloneqq(\pref_1, \ldots, \pref_n))$ be a preference profile.   
  Let~$E\colon \aaa \cup V \to \rz$ be an embedding of the alternatives and the voters into the real line where each two distinct alternatives~$a, b\in \aaa$ have different values, that is,
  $E(a)\neq E(b)$.
  A voter~$v_i \in V$ 
  is \emph{\ode} with respect to~$E$
  if for each two distinct alternatives~$a, b \in \aaa$
  voter~$v_i$ strictly prefers the one closer to him,
  that is,
  \begin{align*}
    \text{if } a \pref_i  b \text{, then } |E(a) - E(v_i)| < |E(b) - E(v_i)|.
  \end{align*}
  
  An embedding~$E$ of the alternatives and voters is a \emph{\oder} of profile~$\ppp$
  if each voter in~$V$ is
  \emph{\ode} with respect to $E$.
  
  A profile is \emph{\ode} if it has a \oder.
\end{definition}

The \ode definition implies the following.

\begin{observation}[\cite{ChePruWoe2017}]
  A voter~$v_i\in \vvv$ is \ode with respect to an embedding $E$ of the alternatives and voters if and only if 
  for each two distinct alternatives~$a$ and $b$ with $a\pref_i b$ it holds that
  if $E(a) < E(b)$ then $E(v_i) < \frac{1}{2}(E(a)+ E(b))$;
  otherwise $E(v_i) > \frac{1}{2}(E(a)+ E(b))$.
\end{observation}

The following observation regarding the relation between single-peaked and single-crossing profiles and the \ode representation is also known from the literature~\cite{Coombs1964,DoiFal1994,ChePruWoe2017}.

\begin{observation}\label{obs:oder_is_sp}
  If a profile is \ode, then it is also single-peaked and single-crossing.
\end{observation} 

\begin{proof}
  It is straight-forward to see that
  if there is a \oder~$E$ of a given profile,
  then this profile is single-peaked with respect to the order induced by ordering the alternatives according to their values in $E$.
  Moreover, it is single-crossing with respect to the order induced by ordering the voters according to their values in $E$.
\end{proof}

\section{Single-peaked profiles with two voters are \ode}
\label{sec:main-result-1}

In this section, we formulate and prove our first main result.

\begin{theorem}\label{thm:Euclidean-relationship}
 Given a profile with two voters~$v_1$ and $v_2$ whose preference orders~$\pref_1$ and $\pref_2$ are single-peaked,
 \cref{alg:Euclid-embedding} returns a \oder of this profile.
\end{theorem}

We can conclude the following from \cref{thm:Euclidean-relationship}.

\begin{theorem}\label{thm:two-votes-SP=Euclidean}
  A profile~$\mathcal{P}$ with two voters is one-dimensional Euclidean if and only if it is single-peaked.
\end{theorem}

\begin{proof}
  The ``only if'' part follows from \cref{obs:oder_is_sp} and
  the ``if'' part follows from \cref{thm:Euclidean-relationship}.
\end{proof}

In the remainder of this section, we show the correctness of \cref{thm:Euclidean-relationship}. 

\subsection{\cref{alg:Euclid-embedding} and some technical results}\label{subsec:alg-technicals}

     \newcommand{\lleft}{\textsf{left}}
     \newcommand{\rright}{\textsf{right}}
     \newcommand{\dist}{\textsf{dist}}
\begin{algorithm}[t!]
     \DontPrintSemicolon
     \footnotesize
   \SetKwInput{KwInput}{Input}
   \SetKwInput{KwOutput}{Output}
   \SetKwBlock{Block}
   \SetAlCapFnt{\footnotesize}
   \KwInput{\\
     $\hspace{3pt}$ $a_1 \pref_1 a_2 \pref_1 \ldots \pref_1 a_m$ --- voter~$v_1$'s preference order
     
     $\hspace{3pt}$ $b_1 \pref_2 b_2 \pref_2 \ldots \pref_2 b_m$ --- voter~$v_2$'s preference order
   }
   \KwOutput{Embeddings~$E\colon \{1,2,\dots, m\} \cup \{v_1,v_2\} \to \mathds{R}$}
   \vspace{3mm}	

   \Embed{$\pref_1$, $\pref_2$}:
   \Block{
     
     \tcc{Initialize the embedding by starting with the voters and the `inner' alternatives.}
     
     Let $p\leftarrow 1$\label{alg:ini-start}
     
     \For{$i=1,2\dots,m$}{\label{alg:inner-start}
       \If{$b_i \succeq_1 b_1 \wedge b_i \succeq_2 a_1$}{
         $E(b_i) = p$
         
         $p \leftarrow p+1$
       }
     }\label{alg:inner-end}

     $E(v_2) = 0$ \label{alg:v2-pos}
     
     $E(v_1) = p$ \label{alg:ini-end}
     
     \Repeat{$E(c)$ defined for all alternative~$c$}{
       \label{alg:loop-start}
       
       $s_1 \leftarrow $ Refine($\pref_1$,$v_1$)\label{alg:refine-v1}

       $s_2 \leftarrow $ Refine($\pref_2$,$v_2$)\label{alg:refine-v2}

       \If{$s_1= \mathsf{false} \wedge s_2= \mathsf{false}$}{
         Fallback() \label{alg:fallback}
       }
     }\label{alg:loop-end}
   }

   \tcc{Refine positions for preference order $\pref$}

  \Refine{$\pref\colon c_1 \pref \cdots \pref c_m$, $v\in \{v_1, v_2\}$}:
  \Block{
   \label{alg:refine-start}

      $j \leftarrow \argmin\{x\in \{1,2,\dots,m\}\mid E(c_x)=\text{undefined}\}$
     
      $i \leftarrow  \argmin\{x\in \{j+1,j+2,\dots,m\}\mid E(c_x)=\text{defined}\}$
     
      \If{$i$ exists}{%
        \label{alg:i-exists-start}
        $\dist({j-1}) \leftarrow |E(v)-E(c_{j-1})|$ \label{alg:dist-j-1}

        $\dist({i}) \leftarrow |E(v)-E(c_{i})|$\label{alg:dist-i}

        \For{$k=j,j+1,\dots,i-1$\label{alg:for-j-i-1}}
        {
          \If{$v=v_1$}{
            $E(c_k)\leftarrow
            E(v_1) + \dist({j-1}) + 
            \frac{\dist({i})-\dist({j-1})}{i-j+1}\cdot (k-j+1)$\label{alg:i-exists-v1}
          }
         \Else{%
           $E(c_k)\leftarrow E(v_2) - \dist({j-1}) - \frac{\dist({i})-\dist({j-1})}{i-j+1}\cdot (k-j+1)$\label{alg:i-exists-v2}

         }
       }
     
       \Return{$\mathsf{true}$}\label{alg:i-exists-end}
     }
     \Else{
       \Return{$\mathsf{false}$}\label{alg:refine-end}
     }
   }

   \tcc{Refine fall back}
   \Fallback{}:
   \Block{
     $j \leftarrow \argmin\{x\in \{1,2,\dots,m\}\mid E(a_x)=\text{undefined}\}$
     \label{alg:refine-ff}%
     
     \If{$j$ exists}{%
       \label{alg:ff-j-start}
       $E(a_j)\leftarrow E(v_1) + |E(v_1) - E(a_{j-1})| + 1 $ \label{alg:ff-embed}
     }\label{alg:ff-j-end}
   }
   
   \caption{Algorithm for computing a \ode embedding for two preference orders, one with $a_1\pref_1 a_2 \pref_1 \dots \pref_1 a_m$ and the other $b_1 \pref_2 b_2 \pref_2 \dots \pref_2 b_m$.}
   \label{alg:Euclid-embedding}
\end{algorithm}

The general idea behind the algorithm in \cref{thm:Euclidean-relationship} is to first embed all \emph{inner} alternatives that are ranked by both voters~$v_1$ and $v_2$ between $a_1$ and $b_1$,
and embed voter~$v_1$ (resp.\ $v_2$)  to the left of alternative~$a_1$ (resp.\ to right of alternative~$b_1$). 
Then, the algorithm extends the embedding by successively
embedding some appropriately selected alternatives to the left or to the right of the already embedded alternatives so as to obtain an extended embedding
with respect to which one of the voters is \ode.
The single-peaked property, according to \cref{lem:two-sp-alpha}, guarantees that the extension by these alternatives remains a \oder for the other voter.

We introduce the following notion.

\begin{definition}[Inner alternatives]
  Let~$\pref_1$ and $\pref_2$ be two preference orders, and let $a_1$ and $b_1$ be the most preferred alternatives of $\pref_1$ and $\pref_2$, respectively.
  The \emph{set of inner alternatives of $\pref_1$ and $\pref_2$},
  denoted as $\inner(\pref_1, \pref_2)$,
  is the set of all alternatives that are ranked between $a_1$ and $b_1$ by both $\pref_1$ and $\pref_2$:
  \begin{align*}
    \inner(\pref_1,\pref_2)\coloneqq \{c\mid c \succeq_1 b_1 \wedge c \succeq_2 a_1\} 
    \text{.}
  \end{align*}
\end{definition}

\begin{example}
  Consider two preference orders~$\pref_1$ and $\pref_2$ with 
  $1\pref_1 2 \pref_1 3 \pref_1 4$ and $3 \pref_2 4 \pref_2 1 \pref_2 2$.
  The set of inner alternatives by $\pref_1$ and $\pref_2$ is $\inner(\pref_1, \pref_2)=\{1,3\}$.
\end{example}

We observe the following properties concerning the inner alternatives of two single-peaked preferences.
\begin{lemma}\label{lem:inner}
  Consider two preference orders~$\pref_1$ and $\pref_2$.
  \begin{enumerate}[(1)]
    \item\label{lem:inner-first-ranked}
    For each~$r\in \{1,2\}$, the most preferred alternative of $\pref_r$ belongs to $\inner(\pref_1,\pref_2)$.
    \item\label{lem:inner-transitive} For each two distinct inner alternatives~$x,y \in \inner(\pref_1,\pref_2)$ and for an arbitrary alternative~$z$ distinct from $x$ and $y$ it holds that if $z\pref_1 x$ and $z\pref_2 x$,
    then $z\in \inner(\pref_1,\pref_2)$.
    \item\label{lem:inner-reverse}  If $\pref_1$ and $\pref_2$ are single-peaked, then 
     for each two distinct inner alternatives~$x, y\in \inner(\pref_1,\pref_2)$ 
    it holds that $x\succ_1 y$ if and only if $y \succ_2 x$.
  \end{enumerate}
\end{lemma}

\begin{proof}
  The first statement follows from the definition of $\inner$.

  As to the second statement, 
  since $x\in \inner(\pref_1,\pref_2)$,
  if $z\succ_1 x$ and $z\succ_2 x$, then
  by the transitivity of preference orders,
  it follows that $z\succeq_1 b_1$ and $z\succeq_2 a_1$.
  By the definition of $\inner$, we immediately have that $z\in \inner(\pref_1,\pref_2)$.

  It remains to show the last statement.
  If $a_1=b_1$, then by the definition of $\inner$
  it holds that $\inner(\succ_1,\succ_2)=\{a_1\}=\{b_1\}$, and the second statement holds immediately since $\inner(\succ_1,\succ_2)$ has only one alternative.
  Thus, let us assume that $a_1 \neq b_1$ so that $|\inner(\pref_1,\pref_2)|\ge 2$.
  Let $\pref_1$ and $\pref_2$ be single-peaked.
  Suppose, for the sake of contradiction, that there are two distinct alternatives~$x,y\in \inner(\pref_1,\pref_2)$ with $x\succ_1 y$ and $x \succ_2 y$---the case with $y \succ_1 x$ and $y \succ_2 x$ works analogously.
  Since $a_1\pref_1 b_1$ and $b_1 \pref_2 a_1$, it follows that $x,y \notin \{a_1,b_1\}$.
  By the definition of $a_1$ and $b_1$ and since $x,y\in \inner(\succ_1,\succ_2)$,
  this implies that $a_1 \pref_1 x \pref_1 y \pref_1 b_1$ and $b_1 \pref_2 x \pref_2 y \pref_2 a_1$---a contradiction to \cref{lem:two-sp-alpha}.
\end{proof}

Now, we are ready to describe \cref{alg:Euclid-embedding}.
It consists of two parts, the initialization and the main loop.
In the initialization, 
we embed all inner alternatives (see Lines \ref{alg:ini-start}--\ref{alg:inner-end})
in such a way that the left-to-right (resp.\ the right-to-left) order corresponds to the preferences of $v_2$ (resp.\ $v_1$).
By \cref{lem:inner}~(\ref{lem:inner-reverse}), it follows that the most preferred alternative of $v_1$, denoted as $a_1$, is the right-most alternative, while
the most preferred alternative of $v_2$, denoted as $b_1$, is the left-most alternative.
Then, in Lines \ref{alg:v2-pos}--\ref{alg:ini-end} we embed voter~$v_2$ (resp.\ $v_1$) to the left of $a_1$ (resp.\ to the right of $b_1$).
Summarizing, we observe the following about the initialization step.

\begin{proposition}\label{prop:alg-ini}
  Let $E$ be the embedding constructed by the end of the initialization phase (Lines~\ref{alg:ini-start}--\ref{alg:ini-end}) of \cref{alg:Euclid-embedding}.
  Let $c_1,\ldots, c_x$ be the embedded alternatives with $E(c_1)<\cdots < E(c_x)$.
  The following holds.
  \begin{enumerate}[(1)]
    \item $\inner(\succ_1,\succ_2)=\{c_1,c_2,\dots,c_x\}$ with $c_x=a_1$ and $c_1=b_1$.
    \item Voter~$v_2$ prefers $c_1 \succ_2 c_2 \succ_2 \dots \succ_2 c_x$.
    \item\label{prop:ini--dist-v2-v1} $E(v_1)-E(v_2)=|\inner(\succ_1,\succ_2)|+1$.
    \item\label{prop:ini-order} $E(v_2) < E(c_1)$ and $E(c_{x})<E(v_1)$.
    \item If $\pref_1$ and $\pref_2$ are single-peaked,
    then voter~$v_1$ prefers $c_x \succ_1 c_{x-1} \succ_1 \dots \succ_1 c_1$.
  \end{enumerate}
\end{proposition}

\begin{proof}
  The first three statements follow directly 
  from Lines~\ref{alg:inner-start}--\ref{alg:inner-end} and from the definition of $\inner(\succ_1,\succ_2)$.
  Moreover, it holds that $E(c_1)=1$, $E(c_x)=|\inner(\succ_1,\succ_2)|$, $E(v_2)=0$, and $E(v_1)=|\inner(\succ_1,\succ_2)|+1$.
  This implies the fourth statement.

  As to the last statement,
  consider an arbitrary embedded alternative~$c_j$ with $j\in \{1,\ldots,x-1\}$.
  Then, by the second statement, we have that $c_j \succ_2 c_{j+1}$.
  By \cref{lem:inner}(\ref{lem:inner-reverse}),
  we have that $c_{j+1} \succ_1 c_j$.
\end{proof}

After having embedded all inner alternatives,
the main loop (Lines \ref{alg:loop-start}--\ref{alg:loop-end}) extends the embedding by alternatingly placing 
alternatives that should be embedded to the right of the existing embedding
and alternatives that should be embedded to the left of the existing embedding.
The corresponding procedure is called Refine() (Lines \ref{alg:refine-start}--\ref{alg:refine-end}) and is used for both voters~$v_1$ and $v_2$.
It searches through the alternatives along the preference order of $v_1$ (resp.\ $v_2$),
finds the first not-yet-embedded alternative(s) that are ranked between two consecutive embedded alternatives by $v_1$ (resp.\ $v_2$),
and embeds them to the right (resp.\ left) of the right-most (resp.\ left-most) alternative. Fallback() in Line \ref{alg:fallback} guarantees that at least one alternative is embedded during each iteration, thus ensuring that the algorithm terminates.

In the following, we prove that our constructed embedding is indeed a \oder by showing that whenever the embedding at the beginning of an iteration of the main loop is a \oder with regard to the already embedded alternatives,
the embedding at the end of the iteration is also a \oder.
To this end, let $D=\{d_1,d_2,\ldots, d_w\}$ be the alternatives that are embedded
at the beginning of an iteration (Line~\ref{alg:loop-start}) with $E(d_1) < E(d_{2}) < \cdots < E(d_w)$,
  and assume that $E$ is a \oder for all alternatives from~$D$.
  We introduce the following notation regarding the concept of a worst alternative among a given set of alternatives.
  Let $\worst(D,v_1)$ (resp.\ $\worst(D,v_2)$) denote
  the alternative from $D$ that is least preferred by voter~$v_1$ (resp.\ $v_2$)  \emph{i.e.}
  \begin{align*}
    \worst(D,v_1) \in D \text{ with } (D\setminus \{a^*\}) \pref_1 a^*, \text{~and~} 
    \worst(D,v_2) \in D \text{ with } (D\setminus \{b^*\}) \pref_2 b^*.
  \end{align*}

  \begin{example}\label{ex:inner}
    Consider the following preferences of voters~$v_1$ and $v_2$.
    \begin{align*}
      v_1 \colon & 1 \succ_1 4 \succ_1 2 \succ_1 3 \succ_1 5 \succ_1 6 \succ_1 7 \succ_1 8,\\
      v_2 \colon & 3 \succ_2 2 \succ_2 1  \succ_2 5 \succ_2 6 \succ_2 4 \succ_2 8 \succ_2 7.
    \end{align*}
    If $D=\{1,2,3,4\}$, then $\worst(D,v_1)=3$ and $\worst(D,v_2)=4$.
  \end{example}

  We introduce another notion called \emph{no later than}.
  \begin{definition}[No later than]
  For two distinct alternatives~$x$ and $y$,
  we say that $x$ is embedded no later than $y$ if one of the following holds.
  \begin{enumerate}[(1)]
    \item Alternatives~$x$ and $y$ are both embedded during initialization.
    \item They are both embedded in the same call to Refine().
    \item When $y$ is to be embedded, $E(x)$ is already defined.
  \end{enumerate}
  \end{definition}

  To show that each iteration (Lines~\ref{alg:loop-start}--\ref{alg:loop-end}) maintains the \ode property of the embedding, we observe the following useful properties.
  \begin{lemma}\label{lem:x-no-later-than-y}
    Let $x$ and $y$ be two distinct alternatives with $x \succ_1 y$ and $x \succ_2 y$.
    Then, \cref{alg:Euclid-embedding} embeds $x$ no later than $y$. 
  \end{lemma}

  \begin{proof}
    If $y \in \inner(\succ_1,\succ_2)$, then by \cref{lem:inner}~(\ref{lem:inner-transitive}),
    it follows that $x\in \inner(\pref_1,\pref_2)$,
    meaning that $x$ and $y$ are both embedded during the initialization,
    and that $x$ is embedded no later than~$y$.

    Now, let us assume that $y \notin\inner(\pref_1,\pref_2)$. 
    Consider the step when $y$ was embedded.
    There are three cases.
    
    If $y$ has been embedded in Line \ref{alg:for-j-i-1} in a call to Refine($\succ_1$, $v_1$), then let $j$ and $i$ be the indices as defined in that call such that $a_j \succeq_1 y \succ_1 a_i$.
    If $E(x)$ was defined, \emph{i.e.}\ $x$ has already been embedded, then by the definition of ``no later than'',~$x$ is embedded no later than $y$.  
    If $E(x)$ was not defined, then since $a_j$ was defined as the first alternative that is not yet embedded,
    it follows that $a_j \succeq_1 x$.
    Since $x \succ_1 y$, it follows that $a_j \succ_1 x \succ_1y \succ_1 a_i$, implying that $x$ is embedded in the same call to Refine~() as $y$.
    Thus, $x$ is embedded no later than $y$.


    Using a reasoning similar to the previous case, we can infer that $x$ is also embedded no later than $y$ when $y$ has been embedded in call to Refine($\succ_2$, $v_2$) because $x\succ_2 y$.
    
    If $y$ has been embedded in the subprocedure Fallback() in Line \ref{alg:fallback},
    meaning that it is also the only alternative that is embedded during that iteration, then
    line \ref{alg:refine-ff} guarantees that  $E(x)$ was already defined, and thus,~$x$ is embedded no later than $y$.
    %
    %
  \end{proof}

  The next lemma ensures that each alternative not from $\inner(\pref_1,\pref_2)$ is embedded by exactly one of the three subroutines---Refine($\pref_1$, $v_1$), Refine($\pref_2$, $v_2$), or Fallback().
  \begin{lemma}\label{lem:only-one-case-applies}
    Let $\pref_1$ and $\pref_2$ be two single-peaked preference orders. 
    Consider an arbitrary not-yet-embedded alternative~$x$, 
    \emph{i.e.}\ $x\notin D$.
    Then, $\worst(D,\succ_1)\pref_1 x$ or $\worst(D,\succ_2) \pref_2 x$.
  \end{lemma}

  \begin{proof}
    Let $a^*=\worst(D,\pref_1)$ and $b^*=\worst(D,\pref_2)$. 
    Towards a contradiction, suppose that $\pref_1$ and $\pref_2$ are single-peaked but
    $x$ is an alternative with $x\notin D$ such that
    \begin{align}
      x\pref_1 a^* \text{ and } x\pref_2 b^*.\label{lem:assum}
    \end{align}
    This implies that
    \begin{align}
      a_1 \neq a^* \text{ and } b_1 \neq b^*,\label{eq:first-rankeds}
    \end{align}
    as $a_1 \succ_1 x$ and $b_1 \succ_2 x$; recall that $a_1$ (resp.\ $b_1$) is the alternative most preferred by voter~$v_1$ (resp.\ $v_2$). 
    By \cref{lem:x-no-later-than-y} and since $x$ is not yet embedded,
    the assumption \eqref{lem:assum} also implies that 
    \begin{align}\label{lem-eq:temp-votes}
      & v_1\colon b^*\succ_1 x \succ_1 a^* \text{ and } v_2\colon a^* \succ_2 x \succ_2 b^*,\text{ and thus, }  \\
     & a^* \neq b^*. \label{eq:worst-different}
      \end{align}
      By the definitions of $a_1$ and $b_1$, 
      we further infer that 
    \begin{align}
      v_1\colon a_1 \succeq_1 b^*\succ_1 x \succ_1 a^* \text{ and } v_2\colon  b_1 \succeq_2 a^* \succ_2 x \succ_2 b^*. \label{lem-eq:current-votes}
    \end{align}
     We distinguish between two cases, in each case aiming to obtain $x\in \inner(\pref_1, \pref_2)$ which is a contradiction to $x\notin D$ as $\inner(\pref_1,\pref_2)\subseteq D$.

     \noindent \textbf{Case~1:} If $a_1 = b^*$, then the preferences given in~\eqref{lem-eq:current-votes} are equivalent to 
       \begin{align}
      v_1\colon a_1 \succ_1 x \succ_1 a^* \text{ and } v_2\colon b_1 \succeq_2 a^* \succ_2 x \succ_2 a_1. \label{lem-eq:current-votes-2}
       \end{align}
       Furthermore, $b_1\neq a^*$ as otherwise $x\in \inner(\pref_1,\pref_2)$---a contradiction.
        Consequently, the preferences given in~\eqref{lem-eq:current-votes-2} imply that 
        \begin{align}
          v_1\colon a_1 \succ_1 x \succ_1 a^* \text{ and } v_2\colon b_1 \succ_2 a^* \succ_2 x \succ_2 a_1. \label{lem-eq:current-votes-3}
        \end{align}
        Since $\pref_1$ and $\pref_2$ are single-peaked,
       by \cref{lem:two-sp-alpha} and by \eqref{lem-eq:current-votes-3},
       we must have that $x\succ_1 b_1$.
       However, this implies that $x\in \inner(\pref_1, \pref_2)$ since $x\succ_2 a_1$---a contradiction.

     \noindent\textbf{Case~2:}  If $a_1 \neq b^*$, then the preferences given in~\eqref{lem-eq:current-votes} imply that 
       \begin{align}
      v_1\colon a_1 \succ_1 b^* \succ_1 x \succ_1 a^* \text{ and } v_2\colon b_1 \succeq_2 a^* \succ_2 x \succ_2 b^*. \label{lem-eq:current-votes-4}
       \end{align}
       Since $\pref_1$ and $\pref_2$ are single-peaked,
       by \cref{lem:two-sp-alpha} and by \eqref{lem-eq:current-votes-4},
       we must have that $x\succ_2 a_1$ and $x\succ_1 b_1$, implying that $x\in \inner(\pref_1, \pref_2)$---a contradiction.
     \end{proof}

     For two distinct alternatives that have not been embedded, we observe the following.

     \begin{lemma}\label{lem:only-one-case-applies-2}
        Let $\pref_1$ and $\pref_2$ be two single-peaked preference orders. 
        Let $x$ and $y$ be two distinct alternatives that have not been embedded,
        \emph{i.e.}\ $x,y\notin D$ with $x\neq y$.
        For each~$r\in \{1,2\}$ it holds that
        if $x\pref_r y \pref_r \worst(D,\succ_r)$,
        then $\worst(D,\succ_s) \pref_s x \pref_s y$, where $s \in \{1,2\}\setminus \{r\}$. 
     \end{lemma}

     \begin{proof}
       \newcommand{\best}{\mathsf{best}}
       Assume that $x \pref_r y \pref_r \worst(D,\succ_r)$ holds.
       Let $\mathsf{best}(\succ_r)$ be the most preferred alternative in the preference order~$\succ_r$.
       Then, we have that
       \begin{align}
         & \mathsf{best}(\succ_r) \pref_r x \pref_r y \pref_r \worst(D, \pref_r). \label{eq:lem:only-one-r}
       \end{align}       
       By \cref{lem:only-one-case-applies},
       it follows that $\worst(D,\succ_s) \pref_s \{x,y\}$.
       Thus, it remains to show that $x\pref_s y$.
       Towards a contradiction, suppose that $y \pref_s x$.
       By the definition of $\worst(D, \pref_s)$, voter~$v_s$ must have preferences
       \begin{align}
         \{\mathsf{best}(\succ_r), \worst(D,\pref_r)\} \succeq_s \worst(D,\succ_s) \pref_s y \pref_s x.\label{eq:only-one-case-applies-2-contrad}
       \end{align}

    \noindent   Together with~\eqref{eq:lem:only-one-r}, we have
       \begin{align*}
         \best(r) \pref_r x \pref_r y \pref_r \worst(D,\succ_r), \text{ and }
         \{\best(r), \worst(D,\succ_r)\} \pref_s y \pref_s x,
       \end{align*}
       a contradiction to \cref{lem:two-sp-alpha}.
     \end{proof}
     
      We observe the following for the subprocedure Refine($\succ$, $v$).

     \begin{lemma}\label{lem:refine-v}
       Let $j$ and $i$ be as defined in a call to~Refine($\succ$, $v$).
       If $E$ is a \oder for the alternatives~$c_{j-1}$ and $c_i$ and for the voter~$v$ %
       so that the 1-Euclidean property is satisfied,
       then 
       $|E(c_{j-1})-E(v)|<|E(c_j)-E(v)|<\dots< |E(c_{i-1})-E(v)|<|E(c_i)-E(v)|$.
     \end{lemma}

     \begin{proof}
       By the definitions of $j$ and $i$,
       since the 1-Euclidean property is maintained for voter~$v$,
       it holds that
       \begin{align}
         \dist(j-1)=|E(c_{j-1})-E(v)| < |E(c_i)-E(v)| = \dist(i), \label{eq:1-Euclid-j-1+i}
       \end{align}
       Note that $\dist(j-1)$ and $\dist(i)$ are defined in Lines \ref{alg:dist-j-1}--\ref{alg:dist-i}.
       From Lines~\ref{alg:i-exists-v1}--\ref{alg:i-exists-v2}, it is straightforward to verify that
       for each $c_k$ with $j-1\le k \le i$,
       \begin{align}
         |E(c_k)-E(v)| & = \dist({j-1}) + \frac{\dist({i})-\dist({j-1})}{i-j+1}\cdot (k-j+1)\label{eq:dist-ck-1}
       \end{align}

       \noindent Combining \eqref{eq:dist-ck-1} with \eqref{eq:1-Euclid-j-1+i}, we obtain
       the chain of inequalities in the lemma.
       %
     \end{proof}

\subsection{Correctness of \cref{thm:Euclidean-relationship}}
From \cref{lem:refine-v}, we know that Refine($\succ$, $v$) ensures the \ode property for voter~$v$ with regard to the already embedded and the newly added alternatives.
Together with \cref{lem:x-no-later-than-y,lem:only-one-case-applies,lem:only-one-case-applies-2}, we are ready to show the correctness of \cref{thm:Euclidean-relationship}.
  

\begin{proof}[Proof of \cref{thm:Euclidean-relationship}]
  Let $\pref_1$ and $\pref_2$ be single-peaked.
  First of all, by \cref{prop:alg-ini}, the initialization phase computes
  a \oder of the two preference orders~$\pref_1$ and $\pref_2$ when restricted to the inner alternatives~$\inner(\succ_1, \succ_2)$.
  Thus, to prove the correctness,
  we only need to show that each iteration of the main loop (Lines~\ref{alg:loop-start}--\ref{alg:loop-end}) returns an extended embedding that is
  a \oder of the embedded alternatives. 
  To achieve this, we need to show that the procedures Refine($\succ_1$, $v_1$), Refine($\succ_2$, $v_2$),
  and Lines~\ref{alg:ff-j-start}--\ref{alg:ff-j-end} extend the existing \oder  
  to one that is \ode with respect to the alternatives that have already been embedded and also with respect to those which are newly embedded.
  To this end, let $D$ be the alternatives that are embedded at the beginning of an iteration, 
  and assume that $E$ is a \oder for all alternatives in~$D$.
  Let $a^*=\worst(D,v_1)$ (resp.\ $b^*=\worst(D,v_2)$) denote  the alternative from $D$ that is least preferred by voter~$v_1$ (resp.\ $v_2$).
  Let $C$ be the set of alternatives that are to be embedded. 

  We distinguish between three cases.

  \noindent \textbf{Case 1:}
  $C$ has been embedded in a call to Refine($\succ_1$, $v_1$).
  By the procedure~Refine($\succ_1$, $v_1$),
  the two embedded alternatives that the algorithm identifies are $a_j$ and $a_i$ such that $C=\{a_j,a_{j+1},\dots,a_{i-1}\}$ and
  \begin{align}
    a_j \succ_1 a_{j+1} \succ_1 \dots a_{i-1} \succ_1 a_i \succeq_1 a^*.\label{eq:case1-v1}
  \end{align}
  By assumption, the embedding~$E$ is a \oder for voters~$v_1$ and $v_2$ and for the alternatives from $D$. 
  By \cref{lem:refine-v},
  it follows that $E$ is also a \oder for voter~$v_1$ and for all alternatives from $D\cup C$.
  In particular, it holds that, 
  \begin{align}
    \forall k \in \{1,2,\ldots, i-1\}\colon |E(a_k)-E(v_1)| < |E(a_{k+1})-E(v_1)|.\label{eq:case1-v1-dists}
  \end{align}
  
  It remains to show that $E$ is also a \oder for voter~$v_2$ and for all alternatives from 
  $D\cup C$.
  Using \cref{lem:only-one-case-applies-2} on the preferences in~\eqref{eq:case1-v1}, voter~$v_2$ must have preferences~$b^*\succ_2 a_j \succ_2 a_{j+1} \succ_2 \dots \succ_2 a_{i-1}$. 
  By the embedding of the alternatives from $C$ (Line~\ref{alg:i-exists-v1}),
  for each alternative~$a_k$ with $j \le k \le i-1$
  it holds that
  \begin{align}
    E(v_2)<E(v_1)<E(a_k)<E(a_{k+1}), \label{eq:order-v2-v1-aj}
  \end{align}
  implying that
  $|E(a_k)-E(v_2)|<|E(a_{k+1}-E(v_2))|$.
  Thus, to show that~$E$ remains a \oder for voter~$v_2$ regarding the alternatives from $D\cup C$, 
  we only need to show that
  $|E(b^*)-E(v_2)|<|E(a_j)-E(v_2)|$.
  Now, if we can show that
  \begin{align}b^* \succeq_1 a_{j-1},\label{eq:case1-b*}\end{align} then
  we can derive that
  \allowdisplaybreaks  \begin{align*}
    |E(b^*)-E(v_2)| & \le |E(b^*)-E(v_1)| +|E(v_1)-E(v_2)|\\
                    &\stackrel{\eqref{eq:case1-b*}}{<} |E(a_{j-1})-E(v_1)| +|E(v_1)-E(v_2)|\\
    & \stackrel{\eqref{eq:case1-v1-dists}}{<} |E(a_{j})-E(v_1)| + |E(v_1)-E(v_2)|\\
    &\stackrel{\eqref{eq:order-v2-v1-aj}}{=} |E(a_{j})-E(v_2)|,
  \end{align*}
  which is what we needed to show.
  
  Thus, the only task remained is to show that
  \eqref{eq:case1-b*} holds. We distinguish between two cases.
  If $b^* \in \inner(\succ_1,\succ_2)$,
  then $b^*\succeq_2 a_1$. By the definition of $b^*$, this implies $b^*=a_1$, and thus $b^* \succeq_1 a_{j-1}$ as $a_1$ is the first alternative in~$\succ_1$.

  If $b^* \notin \inner(\succ_1,\succ_2)$, then $b^*$ was embedded during a previous iteration of the main loop.
  Let us consider this iteration where $b^*$ was embedded.
  Suppose for the sake of contradiction that $a_{j-1}\succ_1 b^*$.
  By the definition of $a_j$, it follows that
  $a_j \succ_1 a_{j+1} \succ_1 \dots \succ_1 a_{i-1}\succ_1 b^*$.
  Since $a_j$ will be embedded later than $b^*$,
  it follows 
  that Refine($\succ_1$, $v_1$) returned false.
  However, Refine($\succ_2$, $v_2$) also returned false since no alternative~$b_{i'}$ exists that is embedded before $b^*$ such that $b^*\succ_2 b_{i'}$.
  Finally, the subprocedure Fallback() in Line~\ref{alg:fallback} could not have applied since there $a_j$ remained unembedded during this iteration but $a_j \succ_1 b^*$. Thus, there is no way $b^*$ could have been embedded --- a contradiction. 
  Summarizing, we have shown that
  $b^* \succeq_1 a_{j-1}$.
  This completes the proof for the first case.

  \noindent \textbf{Case 2:}  $C$ has been embedded in a call to Refine($\succ_1$, $v_1$). The reasoning is very similar to the one for Case~1. 
  The two embedded alternatives identified by the algorithm are $b_j$ and $b_i$ such that $C=\{b_j,b_{j+1},\dots,b_{i-1}\}$ and
  \begin{align}
    b_j \succ_2 b_{j+1} \succ_2 \dots b_{i-1} \succ_2 b_i \succeq_2 b^*.\label{eq:case2-v2}
  \end{align}
  By assumption, the embedding~$E$ is a \oder for voters~$v_1$ and $v_2$ and for the alternatives from $D$. 
  By \cref{lem:refine-v},
  it follows that $E$ is also a \oder for voter~$v_2$ and for all alternatives from $D\cup C$.
  In particular, it holds that, 
  \begin{align}
    \forall k \in \{1,2,\ldots, i-1\}\colon |E(b_k)-E(v_1)| < |E(b_{k+1})-E(v_1)|.\label{eq:case1-v2-dists}
  \end{align}

  It remains to show that $E$ is also a \oder for voter~$v_1$ and for all alternatives from 
  $D\cup C$.
  Using \cref{lem:only-one-case-applies-2} on the preferences in~\eqref{eq:case2-v2}, voter~$v_1$ must have preferences~$a^*\succ_1 b_j \succ_1 b_{j+1} \succ_1 \dots \succ_1 b_{i-1}$. 
  By the embedding of the alternatives from $C$ (Line~\ref{alg:i-exists-v1}),
  for each alternative~$b_k$ with $j \le k \le i-1$
  it holds that
  \begin{align}
    E(b_{k'}) < E(b_{k}) < E(v_2)<E(v_1), \label{eq:order-v2-v1-bj}
  \end{align}
  implying that $|E(b_k)-E(v_1)|<|E(b_{k'}-E(v_1))|$.
  Thus, to show that~$E$ remains a \oder for voter~$v_1$ regarding the alternatives from $D\cup C$, 
  we only need to show that
  $|E(a^*)-E(v_1)|<|E(b_j)-E(v_1)|$.

  Now, if we can show that \begin{align} a^* \succeq_2 b_{j-1},\label{eq:case2-a*}\end{align} then
  we can derive that
  \allowdisplaybreaks    \begin{align*}
    |E(a^*)-E(v_1)| & \le |E(a^*)-E(v_2)| +|E(v_1)-E(v_2)|\\
    &\stackrel{\eqref{eq:case2-a*}}{<} |E(b_{j-1})-E(v_2)| +|E(v_1)-E(v_2)|\\
    &\stackrel{\eqref{eq:case1-v2-dists}}{<} |E(b_{j})-E(v_2)| +|E(v_1)-E(v_2)|\\
    &\stackrel{\eqref{eq:order-v2-v1-bj}}{=} |E(b_{j})-E(v_1)|,
  \end{align*}
  which is what we needed to show.
  
  Thus, the only task remaining is to show that \eqref{eq:case2-a*} holds. We distinguish between two cases.
  If $a^* \in \inner(\succ_1,\succ_2)$,
  then $a^*\succeq_1 b_1$.
  By the definition of $a^*$, this implies $a^*=b_1$, and it follows that $a^* \succeq_2 b_{j-1}$ as $b_1$ is the first alternative in~$\succ_2$.
  
  If $a^* \notin \inner(\succ_1,\succ_2)$, then $a^*$ was embedded during a previous iteration of the main loop.
  Let us consider the iteration where $a^*$ was embedded.
  Suppose for the sake of contradiction that $b_{j-1}\succ_2 a^*$.
  By the definition of $b_j$, it follows that
  $b_j\succ_2 b_{j+1} \succ_2 \dots \succ_2 b_{i-1}\succ_2 a^*$.
  Since $a^*$ is the embedded alternative least preferred one $v_1$,
  it follows that Refine($\succ_1$, $v_1$) returned false.
  However, when Refine($\succ_2$, $v_2$) was called,
  all alternatives preferred to $a^*$ by $v_2$ must not be embedded later than $a^*$---a contradiction since $b_j$ will be embedded later than $a^*$.
  Summarizing, we have shown that
  $a^* \succeq_2 b_{j-1}$. 
  This completes the proof of the second case.
  

  \noindent \textbf{Case 3:}
  $C$ has been embedded in a call to Fallback().
  Thus, it must hold that \begin{align} D \succ_1 C \text{ and } D\succ_2 C.\label{eq:case3}\end{align} 
  We infer that $C=\{a_j\}$ where $j=|D|+1$,
  and that \begin{align}E(v_1)<E(a_j).\label{eq:v1<aj}\end{align}
  To show the \ode property, we only need to show that
  $|E(a^*)-E(v_1)| < |E(a_j)-E(v_1)|$ and $|E(b^*)-E(v_2)|< |E(a_j)-E(v_2)|$.
  By Lines~\ref{alg:refine-ff}--\ref{alg:ff-j-end}, it holds that
  $a^*=a_{j-1}$.
  Thus, we infer that
  \begin{align}
    |E(a^*)-E(v_1)| =|E(a_j)-E(v_1)|-1 <|E(a_j)-E(v_1)|.\label{eq:case3-v1-1Euclid}
  \end{align}
  By the definition of $a^*$ and $b^*$ voter~$v_1$ has preferences
  \begin{align}
    b^* \succeq_1 a^*.\label{eq:case3-v1}
  \end{align}
  Since $E$ is a \oder for the voter~$v_2$ and for the alternatives in $d$, this implies the following.
  \begin{align*}
    |E(b^*)-E(v_2)| & \le |E(b^*)-E(v_1)|+|E(v_1)-E(v_2)| \\
                    &\stackrel{\eqref{eq:case3-v1}}{\le} |E(a^*)-E(v_1)| + |E(v_1)-E(v_2)|\\
                    & \stackrel{\eqref{eq:case3-v1-1Euclid}}{<} |E(a_j)-E(v_1)|+|E(v_1)-E(v_2)|\\
    & \stackrel{\eqref{eq:v1<aj}} = |E(a_j)-E(v_2)|.
  \end{align*}
  To conclude, we have shown that in each case the algorithm extends the embedding so that the resulting embedding is a \oder of both voters and of the alternatives already embedded as well as the newly embedded alternatives. 
  Thus, our algorithm indeed computes a \oder of two voters whose preferences are single-peaked.
\end{proof}

\begin{example}\label{ex:algorithm}
  We illustrate our algorithm using the profile from \cref{ex:inner}, with two voters and eight alternatives:
  \begin{align*}
      v_1 \colon & 1 \succ_1 4 \succ_1 2 \succ_1 3 \succ_1 5 \succ_1 6 \succ_1 7 \succ_1 8,\\
      v_2 \colon & 3 \succ_2 2 \succ_2 1  \succ_2 5 \succ_2 6 \succ_2 4 \succ_2 8 \succ_2 7.
  \end{align*}
  It is single-peaked with respect to the order $\rhd$ with $8 \rhd 6 \rhd 3 \rhd 2 \rhd 1 \rhd 4 \rhd 7$, and also with respect to the reverse of $\rhd$.
  Given this profile as input, our algorithm will return a \ode embedding which is depicted in the following line.
  
  
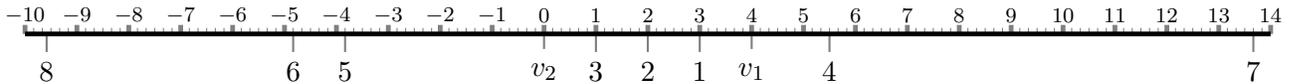
{\centering
\begin{tikzpicture}[>=stealth']
  \def \sscale {.23}
  \def \yy {2.4}
  \def \xsca {3}

  \begin{scope}[xshift=0cm,yshift=0cm]
    \foreach \i / \j \ / \k in {%
      -28.75/8/{-\frac{115}{12}},
      -14.5/6/{-\frac{28}{6}},
      -11.5/5/{-\frac{23}{6}},
      {3}/3/{1},
      {6}/{2}/{2},
      {9}/{1}/{3},
      {16.5}/4/{\frac{11}{2}},
      {41}/7/{\frac{41}{3}}} {
    \node[] at (\i*\sscale, -0.5) (a\j) {$\j$};
    \draw[thick,gray] (\i*\sscale,0) -- (\i*\sscale,-.24);
  }

  \foreach \i / \j  / \k in {0/2/0,12/1/4} {
    \node[] at (\i*\sscale, -0.5) (v\j) {$v_\j$};
    \draw[thick,gray] (\i*\sscale,0) -- (\i*\sscale,-.2);
  }

  \foreach \i    in {-10,...,14} {
    \draw[gray, thickline] (\i*\sscale*\xsca,0) -- (\i*\sscale*\xsca,.12);
    \node[] at (\i*\sscale*\xsca, .24) (n\i) {\scriptsize $\i$};
  }

  \foreach \i    in {-30,-29.5,...,42} {
    \draw[gray] (\i*\sscale,0) -- (\i*\sscale,.08);
  }

  \draw[thickline] (-30*\sscale,0) -- (42*\sscale,0);
  \end{scope}
\end{tikzpicture}
\par }

  First of all, our algorithm embeds the inner alternatives~$\inner(\pref_1,\pref_2)=\{1,2,3\}$ between voter~$v_1$ at $4$ and voter~$v_2$ at $0$.

  In iteration~1, alternative~$4$ is embedded to the right of voter~$v_1$,
  as it is the first not-yet-embedded alternative in the preferences of $v_1$
  and there is an embedded alternative, namely~$2$, such that $v_1$ prefers $4$ to $2$.

  After alternative~$4$ has been embedded, alternatives~$5$ and $6$ are embedded to
  the left of the left-most alternative, namely $3$.
  This is because $v_1$ prefers each embedded alternative to each not-yet-embedded alternative (\emph{i.e.} Refine($\succ$, $v_1$) would return false).
  Alternatives~$5$ and $6$ are the first not-yet-embedded alternatives in the preference order of $v_2$,
  and there is an embedded alternative, namely~$4$, such that $v_2$ prefers $\{5,6\}$ to $4$.

  In iteration 2, that is, after alternatives~$5$ and $6$ have been embedded, 
  neither Refine($\succ_1$, $v_1$) nor Refine($\succ_2$, $v_2$) return true.
  Alternative~$7$ is the first not-yet-embedded alternative in the preference order of $v_1$.
  Thus, the Fallback() function embeds $7$ to the right of $v_1$ so that it becomes the right-most alternative. 

  Finally, in iteration~3, Refine($\succ_1$, $v_1$) returns false.
  Then, in Refine($\succ_2$, $v_2$), alternative~$8$ is embedded to the left of the left-most alternative, namely $6$, 
  as $8$ is the first not-yet-embedded alternative in the preferences of $v_2$
  and there is an embedded alternative, namely~$7$, such that $v_2$ prefers alternative~$8$ to $7$.
  
  The following table summarizes how the algorithm proceeds with the above profile as input.
  More precisely, row one shows the iteration in increasing order; row two shows which subprocedure in the specific iteration has embedded some alternatives (row three).
  For each alternative in row three, the embedding of this alternative is depicted in the last row. 
  \begin{tabular}{|l | p{.5cm}  p{.5cm}  p{.5cm} @{} p{2cm} @{}  p{1.5cm}  @{} p{1.5cm}  @{} p{2cm} @{} p{2cm}  |}
    \hline
    Iteration & \multicolumn{3}{c|}{0} &  \multicolumn{3}{c}{1} & \multicolumn{1}{|c|}{2} & \multicolumn{1}{c|}{3}  \\\hline
    &\multicolumn{3}{c|}{}&\multicolumn{1}{|c|}{}&&\multicolumn{1}{c|}{}&\multicolumn{1}{|c|}{}&\\[-2ex]
    \; Call &\multicolumn{3}{|c|}{Initialization} &  \multicolumn{1}{|c|}{Refine($\pref_1$, $v_1$)} &  \multicolumn{2}{|c|}{Refine($\pref_2$, $v_2$)} &  \multicolumn{1}{|c|}{Fallback()} &  \multicolumn{1}{|c}{Refine($\pref_2$, $v_2$)} \\[-2.6ex]
    &&&&\multicolumn{1}{|c|}{}&\multicolumn{2}{|c|}{}&\multicolumn{1}{|c|}{}&\\\hline
    &\multicolumn{1}{|c|}{}&\multicolumn{1}{|c|}{}&\multicolumn{1}{|c|}{}&\multicolumn{1}{|c|}{}&\multicolumn{1}{|c|}{}&&\multicolumn{1}{|c|}{}&\\[-2ex]
    Embedded alt.\ & \multicolumn{1}{|c|}{1} & \multicolumn{1}{|c|}{2} & \multicolumn{1}{|c|}{3} & \multicolumn{1}{|c|}{4} & \multicolumn{1}{|c|}{5} & \multicolumn{1}{c|}{6} & \multicolumn{1}{|c|}{7} & \multicolumn{1}{c|}{8} \\
    \hline
    &\multicolumn{1}{|c|}{}&\multicolumn{1}{|c|}{}&\multicolumn{1}{|c|}{}&\multicolumn{1}{|c|}{}&\multicolumn{1}{|c|}{}&&\multicolumn{1}{|c|}{}&\\[-2ex]
    Position & \multicolumn{1}{|c|}{$3$} & \multicolumn{1}{|c|}{$2$} & \multicolumn{1}{|c|}{$1$} & \multicolumn{1}{|c|}{$\frac{11}{2}$} & \multicolumn{1}{c|}{$-\frac{23}{6}$} & \multicolumn{1}{c|}{$-\frac{28}{6}$} & \multicolumn{1}{c|}{$\frac{41}{3}$} & \multicolumn{1}{c|}{$-\frac{115}{12}$} \\[-2.6ex]
    &\multicolumn{1}{|c|}{}&\multicolumn{1}{|c|}{}&\multicolumn{1}{|c|}{}&\multicolumn{1}{|c|}{}&\multicolumn{1}{|c|}{}&&\multicolumn{1}{|c|}{}&\\\hline

  \end{tabular}

\end{example}

\section{Single-peaked and single-crossing profiles with up to five alternatives}\label{sec:main-result-2}

In this section, we state and prove our second main result concerning preference profiles with up to five alternatives.

\begin{theorem}\label{thm:5-alts-sp+sc=euclid}
  Each preference profile with up to five alternatives is \ode if and only if it is single-peaked and single-crossing.
\end{theorem}

\begin{proof}
  As already discussed in the introduction, a \ode profile necessarily is single-peaked and single-crossing.
  Thus, to show the theorem, it suffices to show that every single-peaked and single-crosssing preference profile with up to five alternatives is also \ode.
  We achieve this by using a computer program that exhaustively searches for all possible single-peaked and single-crossing profiles with up to five alternatives and provide a \ode embedding for each of them.
  We did some optimization to shrink our search space extensively.
  First of all, we only consider profiles with at least two alternatives and at least two voters who have pairwise \emph{distinct} preference orders as two voters with the same preference order can be embedded at the same position without losing the \ode property.
  Since the relevant profiles in consideration must be single-crossing, by \cite[Lemma 1]{DoiFal1994} and \cite[Section~2.1]{BreCheWoe2013a}, our program only searches for profiles with at most $\binom{m}{2}$ distinct preference orders, where $m$ is the number of alternatives, $3\le m\le 5$.
  The minimum number of voters we need to consider is three as by \cref{thm:two-votes-SP=Euclidean} all single-peaked and single-crossing preference profiles with two voters are \ode.
  
  Second, we assume that one of the preference orders in the sought profile is~$1\succ 2 \succ \ldots \succ m$.
  We denote this order as the canonical preference order.

  Third, using the monotonicity of the single-peaked property, we consider adding a preference order (there are $m!-1$ many) to form a potential relevant single-peaked and single-crossing profile only if it is single-peaked with the canonical one. 
   By \cite[Theorem~12(i)]{LackLack2017},  
  among all $m!-1$ preference orders other than the canonical one,
  there are $\binom{2m-2}{m-1}-1$ preference orders that are single-peaked with the canonical one.
  Note that for $m=5$, the number of potentially single-peaked profiles with $n=\binom{5}{2}+1=11$~voters would be reduced from $\binom{m!-1}{n}=\binom{119}{11}$ to $\binom{\binom{2m-2}{m-1}-1}{11}=\binom{69}{11}$.

  We summarize the number of single-peaked and single-crossing profile with up to $m=5$ alternatives and up to $n=\binom{m}{2}+1$ in \cref{tab:profiles-m=5}.
  Note that we include profiles which have two voters although by \cref{thm:Euclidean-relationship} all single-peaked and single-crossing preference profile with two voters are \ode.

  \begin{table}
    \centering
    \begin{tabular}{ | c | p{0.8cm} p{0.8cm} p{0.8cm} p{0.8cm} p{0.8cm} p{0.8cm} p{0.8cm} p{0.8cm} p{0.8cm} p{0.8cm}p{0.8cm}@{}|}
      \hline
      \diagbox{m}{n} & \multicolumn{1}{c}{2}  & \multicolumn{1}{c}{3} & \multicolumn{1}{c}{4} & \multicolumn{1}{c}{5} & \multicolumn{1}{c}{6} & \multicolumn{1}{c}{7} & \multicolumn{1}{c}{8} & \multicolumn{1}{c}{9} & \multicolumn{1}{c}{10} & \multicolumn{1}{c|}{11}\\
      \hline
      3 & \multicolumn{1}{c}{5}& \multicolumn{1}{c}{6} & \multicolumn{1}{c}{2} & \multicolumn{1}{c}{-} & \multicolumn{1}{c}{-} & \multicolumn{1}{c}{-} & \multicolumn{1}{c}{-} & \multicolumn{1}{c}{-} & \multicolumn{1}{c}{-} & \multicolumn{1}{c|}{-} \\
      4 & \multicolumn{1}{c}{19}& \multicolumn{1}{c}{69} & \multicolumn{1}{c}{108} & \multicolumn{1}{c}{90} & \multicolumn{1}{c}{39} & \multicolumn{1}{c}{7} & \multicolumn{1}{c}{-} & \multicolumn{1}{c}{-} & \multicolumn{1}{c}{-} & \multicolumn{1}{c|}{-}\\
      5 &   \multicolumn{1}{c}{69} &   \multicolumn{1}{c}{567} &   \multicolumn{1}{c}{2124} & \multicolumn{1}{c}{4810} & \multicolumn{1}{c}{7185} & \multicolumn{1}{c}{7273} & \multicolumn{1}{c}{4969} & \multicolumn{1}{c}{2196} & \multicolumn{1}{c}{570}  & \multicolumn{1}{c|}{66}\\
      \hline
    \end{tabular}
    \caption{For each number~$m$ of alternatives stated in the first column and for each number~$n$ of voters stated in the first row, $3 \le m \le 5$ and $2 \le n \le \binom{m}{2}+1$, we summarize the number of single-peaked and single-crossing preference profiles we have produced that contain the canonical preference order $1\succ 2\succ \cdots \succ m$ and no two voters that have the same preference orders. For instance, when $m=3$ and $n=4$, the number of sought preference profiles is $2$, as indicated in row two and column four.}\label{tab:profiles-m=5}
  \end{table}
  We implemented a program which, for each of these produced profiles, uses the IBM ILOG CPLEX optimization software package to check and find a \ode embedding. 
  The results can be found in \url{https://tubcloud.tu-berlin.de/s/ArdQzFd8J6L5YFN} and can be verified via the program given in~\url{https://tubcloud.tu-berlin.de/s/rSNKkm8dtPkRKnE}.
\end{proof}

\section{Conclusion and outlook}
We have shown that
for profiles with at most five alternatives or at most two voters,
being single-peaked and single-crossing suffices for being \ode.
Our research leads to some interesting follow up questions.
First of all, using our computer program from \cref{sec:main-result-2}
we can produce all single-peaked and single-crossing profiles and all \ode profiles.
A natural question is to count the number of structured (\emph{e.g.} single-peaked, single-crossing, \ode) preference profiles and provide a closed formula in terms of the number~$m$ of alternatives and the number~$n$ of voters,
in a similar spirit as recent work by \citeauthor{LackLack2017}~\cite{LackLack2017} and \citeauthor{CheFin2018}~\cite{CheFin2018}.

Second, both the single-peaked and the single-crossing property can be characterized by a few small forbidden subprofiles~\cite{BaHa2011,BreCheWoe2013a}.
However, this is not the case for the \ode property~\cite{ChePruWoe2017}.
Thus we ask: is it possible to characterize \emph{small} \ode preference profiles via a few forbidden subprofiles?
\citeauthor{Chen2016}~\cite[Chapter~4.11]{Chen2016} provided a generic construction and showed that there are at least $n!$ single-peaked and single-crossing profiles with $n=m/2$~voters and $m$~alternatives that are not \ode.
For $m=6$, this number would be $6$.
However, through our computer program we found that for $m=6$ and $n=3$, out of $4179$ single-peaked and single-crossing preference profiles,
there are $48$ ones which are \emph{not} \ode, which is more than $6$.

Last but not least, for $d\ge 2$, $d$-dimensional Euclidean profiles are not necessarily single-peaked nor single-crossing~\cite{BogLas2007}.
In other words, the forbidden subprofiles that are used to characterize single-peaked or single-crossing profiles are not of use to characterize $d$-dimensional Euclidean profiles.
This leads to the question of sufficient and necessary conditions for profiles to be $d$-dimensional Euclidean.
\citet{BogLas2007} answered this question for profiles that may contain ties.
\citet{BulChe2018} used a computer program to verify that all preference profiles with up to seven alternatives and up to three voters are $2$-dimensional Euclidean.

\section*{Acknowledgment}

We thank Laurent Bulteau (Laboratoire d'Informatique Gaspard Monge in Marne-la-Vall{\'e}e, France) for his insight and helpful comments on this work while Jiehua Chen was visiting him in March 2016; the visit was funded by Laboratoire d'Informatique Gaspard Monge in Marne-la-Vall{\'e}e, France.

\newcommand{\bibremark}[1]{\marginpar{\tiny\bf#1}}

\end{document}